\documentclass[runningheads]{llncs}
\usepackage[T1]{fontenc}

\usepackage{amsmath, amsthm, amssymb, graphicx, xcolor, soul, listings, enumerate, url}
\newtheorem*{maintheorem}{Theorem \ref{singleDetLem}}

\begin{document}
\title{Factorisability of Low Dimensional Non-Negative Integer Matrices}

\author{Paul C.~Bell\inst{1}\orcidID{0000-0003-2620-635X} \and
Eva Foster \inst{2}\orcidID{0009-0002-7780-7331} \and
Daniel Reidenbach\inst{2}\orcidID{0000-0001-7996-5291} \and Pavel Semukhin\inst{1}\orcidID{0000-0002-7547-6391}}
\authorrunning{P.~C.~Bell et al.}
\institute{School of Computer Science and Mathematics, Liverpool John Moores University, UK \email{p.c.bell@ljmu.ac.uk}, \email{p.semukhin@ljmu.ac.uk}\and
School of Computing and Mathematical Sciences, Birkbeck, University of London, UK \email{efoste07@student.bbk.ac.uk}, \email{d.reidenbach@bbk.ac.uk}}

\maketitle
\begin{abstract}
We consider the problem of determining if a given two-dime\-nsional nonnegative integer matrix \(M\) is the product of two such matrices, excluding trivial units. A matrix \(M\) with no such factorisation is called \emph{prime} and therefore belongs to the minimal (infinite rank) generator of $2 \times 2$ matrices over the natural numbers, otherwise it is called composite. We also consider the problem of finding a (non-unique) factorisation of a composite matrix. Our results have applications in computational group theory and the theory of codes, where such matrices are called \emph{incidence matrices}. We analyse the complexity of primality and finding a factorisation for a composite matrix, providing a first efficient algorithm.
\end{abstract}

\section{Introduction -- Primality of nonnegative integer matrices}

In this paper we study the semigroup of nonnegative integer matrices and our main consideration is the problem of determining for a given nonsingular matrix \(A \in \mathbb{N}^{n \times n}\) if there exist \(B, C \in \mathbb{N}^{n \times n} \setminus \Pi_n\) such that \(A = BC\), where \(\Pi_n\) is the set of \emph{units} or invertible elements within \(\mathbb{N}^{n \times n}\). If there do not exist such \(B, C\), then we call \(A\) prime, otherwise it is composite\footnote{As we later show, singular matrices are always composite.}. In the case of a composite matrix, we also study the problem of finding a factorisation. Such factorisations are not unique in general. Note that we assume \(\mathbb{N} = {0, 1, 2, \ldots}\) throughout, i.e., \(0\) is a natural number. We thus see that prime matrices appear somewhat like `primes' within the structure of \(\mathbb{N}^{n \times n}\). 

The problem of determining if a given integer is prime was shown to be decidable in polynomial time in a celebrated paper of Agrawal, Kayal, and Saxena \cite{agrawal2004primes}. The techniques used seem to not be transferable to our setting of factorisation over \(\mathbb{N}^{2 \times 2}\) since we have a non-abelian structure. The rationale for ignoring units \(\Pi_n\) in our setting is clear; otherwise every matrix would be composite (similar to ignoring \(\pm 1\) in factorisations over the integers).

Our primary concern in this paper will be the semigroup \(\mathbb{N}^{2 \times 2}\); we therefore note that \(\Pi_2 = \{J, I\}\), where \(J = \big(\begin{smallmatrix} 0 & 1 \\ 1 & 0 \end{smallmatrix}\big)\) and \(I\) is the \(2 \times 2\) identity matrix. It is not difficult to see that the elements of \(\Pi_2\) are the only units within \(\mathbb{N}^{2 \times 2}\) \cite{BCH22}. Consider the following motivating example. We highlight the determinant of the given factorisation since it plays such an important part of our later analysis.

\begin{example} \label{newEx}
Let \(A = \big(\begin{smallmatrix}  1382 & 1243 \\ 1045 & 1316 \end{smallmatrix}\big)\), \(B = \big(\begin{smallmatrix}  26 & 11 \\ 7 & 31 \end{smallmatrix}\big)\), and \(C = \big(\begin{smallmatrix}  43 & 33 \\ 24 & 35 \end{smallmatrix}\big)\). We observe that \(\det(B) = 3^6\) and \(\det(C) = 23\cdot 31\). Moreover, \(A = BC\) and thus \(A\) is composite. One may computationally verify that \(B\) is prime, but \(C\) has several factorisations, e.g., \(C = \big(\begin{smallmatrix}  20 & 1 \\ 9 & 2 \end{smallmatrix}\big) \cdot \big(\begin{smallmatrix}  2 & 1 \\ 3 & 13 \end{smallmatrix}\big) = \big(\begin{smallmatrix}  16 & 1 \\ 1 & 2 \end{smallmatrix}\big) \cdot \big(\begin{smallmatrix}  1 & 0 \\ 5 & 1 \end{smallmatrix}\big) \cdot \big(\begin{smallmatrix}  2 & 1 \\ 1 & 12 \end{smallmatrix}\big)\).
\end{example}

Our aim then is to consider the problem of determining whether some matrix such as \(A\) in Example~\ref{newEx} above, is prime (primality) and finding a factorisation if it is composite. Clearly the primality problem belongs to co-NP, since we can trivially guess the factorisation of \(A\) in linear time. We note that related problems have been studied before and some partial results are known. Notably, it was recently shown that given a finite set of matrices \(\mathcal{G}\), determining if there exists a nonnegative matrix \(M \in \langle \mathcal{G} \rangle\) is decidable, subject to Schanuel's conjecture, if all matrices within \(\mathcal{G}\) commute, but undecidable in general \cite{CO24}. Note that \(\langle \mathcal{G} \rangle\) denotes the semigroup of matrices generated by \(\mathcal{G}\). More relevant to our study, a characterisation was recently given of some basic infinite families of prime matrices within \(\mathbb{N}^{2 \times 2}\) \cite{BCH22,Heil23}, although there such matrices are called \emph{atoms}. As the authors there note, the structure of this semigroup is far from trivial and a more general classification seems difficult. The classification given there provides some necessary conditions for determining if an element is prime and several infinite families of primes that we will recap in the next section. Note that the related and well-studied \emph{Nonnegative Matrix Factorisation (NMF) problem} has a long history, with connections to quantum mechanics, probability theory, polyhedral combinatorics, etc. \cite{AG12}, however NMF is concerned with \emph{real} matrices and therefore is not directly applicable to our setting.

\subsection{Related results on matrix semigroups}

There has been a great deal of interest recently in identifying reachability and factorisation properties of low-dimensional matrix semigroups. These problems are directly related to the automated verification of systems, with connections to linear recurrence sequences, loop termination \cite{KL22,OW15}, formal power series, and weighted automata \cite{BalleM15,Halava2004}. 

Many such problems are undecidable starting from dimension $3$, with several open problems at dimension $2$. The general setting of such problems is that we are given some (usually finite) \emph{generator} \(\mathcal{G} \subseteq \mathcal{F}^{n \times n}\) for some semiring \(\mathcal{F}\) in dimension \(n \geq 1\), and then consider questions about the semigroup \(\mathcal{S}\) \emph{generated} by \(\mathcal{G}\), denoted \(\mathcal{S} = \langle\mathcal{G}\rangle\). A fundamental problem is \emph{membership}: given \(\mathcal{G} \subseteq \mathcal{F}^{n \times n}\) and a target matrix \(M \in \mathcal{F}^{n \times n}\), determine if \(M \in \langle \mathcal{G} \rangle\). The membership problem was shown undecidable by Markov in the 1940s \cite{Ma47} and there have been a number of  advances for variants of this problem \cite{BP10,BHP24,BT97b,Pa70}.

Another related problem is the \emph{matrix freeness} problem: given a finite set of matrices \(\mathcal{G} \subseteq \mathbb{Z}^{n \times n}\), is every matrix within \(\langle \mathcal{G} \rangle\) uniquely factorisable over \(\mathcal{G}\)? This problem is known to be undecidable in dimension three, even for upper-triangular matrices \cite{CHK99,BT97b}, see \cite{CN12} for a good survey. Note that nonnegative matrices in particular are studied in relation to weighted automata \cite{Halava2004}, coding theory \cite{Reu09}, and machine learning \cite{NMFbook}. A recent work considered the \emph{incidence matrix} of a given morphism from the free monoid into itself, characterising all upper-triangular incidence matrices for commutative binary morphisms \cite{Honkala24}.

\subsection{Our contribution}

We note that the main decision problem that we consider (primality) is trivially decidable. Since our matrices are nonnegative, any product of such matrices will grow in size for any reasonable submultiplicative matrix norm. Thus we can use brute force to determine if some matrix is prime, or else to find a factorisation if composite. As we show however, such a brute force technique is inherently slow. 

We provide a more efficient algorithm. We initially proceed by identifying some infinite classes of matrices in \(\mathbb{N}^{2 \times 2}\) which are composite. For the remaining matrices in \(\mathbb{N}^{2 \times 2}\), we  utilise the Smith normal form of the matrices to derive some partial information about potential factorisations.

If we are given \(A = \big(\begin{smallmatrix} a_1 & a_2 \\ a_3 & a_4 \end{smallmatrix}\big) \in \mathbb{N}^{2 \times 2}\) and want to determine if there exist \(B, C \in \mathbb{N}^{2 \times 2} \setminus \{J, I\}\), where \(J = \big(\begin{smallmatrix} 0 & 1 \\ 1 & 0 \end{smallmatrix}\big)\), such that \(A = BC\), then we may note that we essentially have four unknowns in the elements of matrices \(C\); once such a matrix is fixed then we can determine if a matrix \(B \in \mathbb{N}^{2 \times 2}\) exists, satisfying \(A = BC\). Since we know \(\det(A)\), we might nondeterministically guess the determinant \(\det(C)\) (and thus \(\det(B)\)), which then reduces the number of unknowns to three for any such guess. Each of these unknowns is bounded by the maximal element of matrix \(A\). A brute force approach would now require \(\tilde{O}((\max_{1 \leq j \leq 4}{a_j})^3)\) time complexity to determine if \(A\) is prime and to find a factorisation if composite (\(\tilde{O}\) notation is explained in Definition~\ref{def:softO}), see Proposition~\ref{bruteForceProp}. The derivation of the Smith normal form that we use, alongside our analysis, allows us to significantly reduce the algorithmic complexity to \(\tilde{O}(\mu(A))\), where \(\mu(A)\) denotes the \emph{minimal} element of matrix \(A\), giving the main result of our paper which we now state. Note that set \(\Xi\) is defined in Section~\ref{notSec} and captures the most interesting class of \(2 \times 2\) matrices under consideration.

\begin{maintheorem}
    Given a matrix \(A = \big(\begin{smallmatrix} a_1 & a_2 \\ a_3 & a_4 \end{smallmatrix}\big) \in \Xi\) and \(c \in \mathbb{N}\) such that \(c\, | \det(A)\), we can determine  whether there exist non-unit matrices \(B, C \in \mathbb{N}^{2 \times 2}\) such that \(\det(C) = c\) and \(A=BC\), and we can compute such (non-unique) \(B, C\) if they exist, in time \(\tilde{O}(\mu(A))\).
\end{maintheorem}

Note that our algorithm is \emph{linear} in the minimal element of \(A\), whereas the brute force approach is \emph{cubic} in the maximal element of \(A\), however both are exponential in the binary representation of \(A\). We do not currently know of a lower bound for this problem.

\section{Notation and preliminary results}\label{notSec}

We denote by \(\mathbb{F}^{n \times n}\) the set of square \(n \times n\) matrices over some semiring  \(\mathbb{F}\). We will primarily be interested in \(\mathbb{N}^{2 \times 2}\). By \(\mathbb{N}\) we denote the set of all nonnegative integers, i.e., \(\mathbb{N} = {0, 1, 2, \ldots}\)

We will quite extensively rely on Smith normal forms of \(2\times 2\) integer matrices that are explained in the following proposition.

\begin{proposition}[Smith normal form \cite{KannanB79}]\label{SNFprop}
    For any matrix \(A = \big(\begin{smallmatrix} a_1 & a_2 \\ a_3 & a_4 \end{smallmatrix}\big) \in \mathbb{Z}^{2\times2}\), there exists a factorisation \(A = U_A^{-1}S_AV_A^{-1}\), called a \emph{Smith normal form} of \(A\), such that \(U_A\), \(V_A\) are integer matrices with \(\det(U_A) = \det(V_A) = 1\) and \(S_A\) is a diagonal matrix of the form  \(S_A = \big(\begin{smallmatrix} p & 0 \\ 0 & pq \end{smallmatrix}\big)\), where \(p = \gcd(a_1,a_2,a_3,a_4)\) and \(q = \det(A)/p^2\).
\end{proposition}

In our paper, we primarily consider matrices \(A\) satisfying \(\gcd(a_1,a_2,a_3,a_4) = 1\), in which case the matrix \(S_A\) in the Smith normal form will be equal to \(S_A = \big(\begin{smallmatrix} 1 & 0 \\ 0 & \det(A) \end{smallmatrix}\big)\).

Let \(\Pi_n \subseteq \mathbb{N}^{n \times n}\) denote the set of all units within \(\mathbb{N}^{n \times n}\). Define \(\Phi_n \subseteq \mathbb{N}^{n \times n}\) as the set of \emph{primes} of \(\mathbb{N}^{n \times n}\); i.e., \(A \in \Phi_n\) implies that there does not exist \(B, C \in \mathbb{N}^{n \times n} \setminus \Pi_n\) such that \(A=BC\). We may thus call such a matrix \(A\) \emph{prime}. We may then consider the following basic proposition.

\begin{proposition}\label{charProp}
For all \(n \geq 2\), it holds that:
\begin{enumerate}
    \item \(\Phi_n = \mathbb{N}^{n \times n} \setminus (\mathbb{N}^{n \times n} \setminus \{\Pi_n\})(\mathbb{N}^{n \times n} \setminus \{\Pi_n\})\);
    \item \(\Phi_n = \Pi_n\Phi_n\Pi_n\).
\end{enumerate}
\end{proposition}

\begin{proof}
     The first property holds by definition, since \(\Phi_n\) is the set of elements from \(\mathbb{N}^{n \times n}\) which are not the product of two non-units from \(\mathbb{N}^{n \times n}\); i.e., exactly \((\mathbb{N}^{n \times n} \setminus \{\Pi_n\})(\mathbb{N}^{n \times n} \setminus \{\Pi_n\})\). 
    
    To show the second property, let the matrices $A, U_1, U_2\in\mathbb{N}^{n\times n}$ satisfy $A\in\Phi_n$ and $U_1, U_2\in\Pi_n$. Suppose for contradiction that $U_1AU_2\notin\Phi_n$ -- that is, there exist matrices $B, C\in\mathbb{N}^{n\times n}$ where $B, C\notin\Pi_n$ such that $U_1AU_2=BC$. Note that if $U_1, U_2\in\Pi_n$, then $U_1^{-1}, U_2^{-1}\in\Pi_n$, and so $A=U_1^{-1}BCU_2^{-1}$. Then, it follows that $U_1^{-1}B\notin\Pi_n$, otherwise $B=U_1U_1^{-1}B$ is a product of units, contradicting our assumption that $B\notin\Pi_n$ and $CU_2^{-1}$ is not a unit, otherwise $C=CU_2^{-1}U$ would be a product of units, contradicting our assumption that $C\notin\Pi_n$. Therefore, $A$ is a product of two non-units, contradicting our assumption that $A$ is prime, so it must follow that $U_1AU_2$ is prime.
\end{proof}

By Proposition~\ref{charProp}, we thus see that if a matrix \(M \in \mathbb{N}^{n \times n}\) is prime, then so is \(P_1 M P_2\), where \(P_1, P_2 \in \Pi_n\). This motivates the following definition. 

\noindent {\bf Associate matrices}. Given a matrix \(M \in \mathbb{N}^{n \times n}\), then the set of matrices \(\Pi_n M \Pi_n\) is called the set of \emph{associates} of matrix \(M\); and thus \(M'\) is an associate of \(M\) if \(M' \in \Pi_n M \Pi_n\). For example, if \(n=2\), then a matrix \(A = \big(\begin{smallmatrix} a_1 & a_2 \\ a_3 & a_4 \end{smallmatrix}\big)\) has four associates:
\begin{eqnarray}\label{assocEqn}
\left\{A, JA, AJ, JAJ\right\} =
\left\{\begin{pmatrix} a_1 & a_2 \\ a_3 & a_4 \end{pmatrix}, \begin{pmatrix} a_3 & a_4 \\ a_1 & a_2 \end{pmatrix}, \begin{pmatrix} a_2 & a_1 \\ a_4 & a_3 \end{pmatrix}, \begin{pmatrix} a_4 & a_3 \\ a_2 & a_1 \end{pmatrix} \right\}.
\end{eqnarray}

Note that a matrix is prime if and only if all of its associates are. Define \(\mu:\mathbb{N}^{2 \times 2} \to \mathbb{N}\) as the smallest element of a matrix. We will often later assume that \(a_3\) is thus the smallest element of a particular matrix to ease notations.

In the next section, we deal with certain families of matrices for which we know whether they are prime. With that in mind, we now define some useful classes of matrices, namely the notions of dominated, adj-coprime, and \(\Xi\)-matrices. 

\noindent {\bf Dominated matrices.} Matrix \(A = \big(\begin{smallmatrix} a_1 & a_2 \\ a_3 & a_4 \end{smallmatrix}\big)\) is \emph{dominated} if, up to associate (multiplying by the unit \(J = \big(\begin{smallmatrix} 0 & 1 \\ 1 & 0 \end{smallmatrix}\big)\) as necessary so that \(\mu(A) = a_3\)), it holds that \(a_4 \leq a_2\) or \(a_1 \leq a_2\). Note that in this case, each element of one row or column of the matrix is less than or equal to the other. If \(A\) is not dominated, we call it \emph{non-dominated}. Note that if a matrix \(A\) is dominated, then necessarily \(\big(\begin{smallmatrix} 1 & 1 \\ 0 & 1 \end{smallmatrix}\big)\) is a factor \cite{BCH22}, which is why this is a useful notion. 

\noindent {\bf Adj-coprime matrices.} If all adjacent elements of \(A = \big(\begin{smallmatrix} a_1 & a_2 \\ a_3 & a_4 \end{smallmatrix}\big)\) are coprime (i.e., \(\gcd(a_1a_4, a_2a_3) = 1\)), then we call \(A\) \emph{adj-coprime}; otherwise it is \emph{non adj-coprime}. A matrix is non-dominated or adj-coprime if and only if all of its associates are, as seen from Equation~(\ref{assocEqn}). 

\noindent {\bf Set \(\Xi\).} Let \(\Xi \subseteq \mathbb{N}^{2 \times 2}\) denote the set of all nonnegative \(2\times 2\) integer matrices that are both non-dominated and adj-coprime. We then call a matrix \(X \in \Xi\) a \emph{\(\Xi\)-matrix}.
\begin{example}
    Consider the matrices $A=\big(\begin{smallmatrix}
        10 & 9 \\ 2 & 12
    \end{smallmatrix}\big), B=\big( \begin{smallmatrix}
        17 & 15 \\ 12 & 13
    \end{smallmatrix} \big)$ and $C=\big(\begin{smallmatrix}
        8 & 7 \\ 5 & 12
    \end{smallmatrix}\big)$. Notice that $A$ is non-adj-coprime and non-dominated. Notice that $B$ is adj-coprime and dominated. Finally, $C$ is adj-coprime and non-dominated. As $C$ is adj-coprime and non-dominated, then $C$ is a $\Xi$-matrix.
\end{example}
We may also observe that in Example~\ref{newEx}, \(A\), \(B\), and  \(C\) are \(\Xi\)-matrices.

We will use \emph{soft-O} notation \(\tilde{O}\) which hides polylogarithmic factors when discussing the complexity of our algorithms.
\begin{definition}\label{def:softO}
Let \(f: \mathbb{Z}^{2\times 2}\to \mathbb{R}\) be a nonnegative function on integer matrices \(A = \big(\begin{smallmatrix} a_1 & a_2 \\ a_3 & a_4 \end{smallmatrix}\big)\). Then \(\tilde{O}(f(A))\) will be used as a shorthand for \(O(f(A)\cdot \log^k\max_{1\leq j \leq 4}|a_j|)\) for some \(k>0\).    
\end{definition}
For example, in Theorem~\ref{singleDetLem}, the \(\tilde{O}(\mu(A))\) time complexity comes from our procedure requiring \(\mu(A)\) iterations, and each iteration can be done in time polynomial in the binary representation of \(A\), that is, in time \(O(\log^k\max_{1\leq j \leq 4}|a_j|)\) for some \(k>0\). 

\section{Primality for simple matrix families}

We now collect some results concerning families of matrices from \(\mathbb{N}^{2 \times 2}\) that are known to be prime. These results are either new or come from \cite{BCH22,Heil23,RM86}.

An early work in this area showed that \(2 \times 2\) matrices with determinant $1$ are almost always composite.

\begin{theorem}[\cite{RM86}]\label{detOneAtoms}
Let \(S = \{M | M \in \mathbb{N}^{2 \times 2} \text{ and } \det(M) = 1\}\). Then \(
P = \big(\begin{smallmatrix} 1 & 0 \\ 1 & 1 \end{smallmatrix}\big) \text{ and }
Q = \big(\begin{smallmatrix} 1 & 1 \\ 0 & 1 \end{smallmatrix}\big)
\) 
are the only primes within \(S\).
\end{theorem}

Matrices in \(\Xi\) with a smallest entry no larger than $3$ can be classified as shown in the next proposition.

\begin{proposition}[\cite{BCH22}]
Let \(A \in \Xi\). If \(\mu(A) \leq 3\), then \(A\) is prime.
\end{proposition}

The following proposition (Proposition~4.5 of \cite{BCH22}, rephrased here using our terminology) identifies a large family of composite matrices within \(\mathbb{N}^{2 \times 2}\).

\begin{proposition}[\cite{BCH22}]\label{BaethXiprop}
Let \(A = \big(\begin{smallmatrix} a & b \\ c & d \end{smallmatrix}\big) \in \mathbb{N}^{2 \times 2}\) be prime. Then
\begin{itemize}
\item At most two of \(a, b, c, d\) can be zero;
\item If one or two \(a, b, c, d\) are zero then \(A\) is  from the following set (where \(p \in \mathbb{N}\) is prime):
\[
\left\{\begin{pmatrix} p & 0 \\ 0 & 1 \end{pmatrix},  \begin{pmatrix} 0 & p \\ 1 & 0 \end{pmatrix}, \begin{pmatrix} 0 & 1 \\ p & 0 \end{pmatrix}, \begin{pmatrix} 1 & 0 \\ 0 & p \end{pmatrix}, \begin{pmatrix} 1 & 1 \\ 0 & 1 \end{pmatrix}, \begin{pmatrix} 0 & 1 \\ 1 & 1 \end{pmatrix}, \begin{pmatrix} 1 & 1 \\ 1 & 0 \end{pmatrix}, \begin{pmatrix} 1 & 0 \\ 1 & 1 \end{pmatrix} \right\};
\]
\item If none of \(a, b, c, d\) are zero, then \(A \in \Xi\). 
\end{itemize}
\end{proposition}

Note that the contrapositive of Proposition~\ref{BaethXiprop} proves that \(A \not\in \Xi\) implies that \(A\) is composite. The following proposition was proven in \cite{BCH22} as Proposition~$4.8$.

\begin{proposition}[\cite{BCH22}]\label{bchprop}
Let \(A = \big(\begin{smallmatrix} x & x+1 \\ x+1 & x \end{smallmatrix}\big) \in \mathbb{N}^{2\times 2}\) be a nonunit. Then \(A\) is prime if and only if \(2x+1\) is prime.
\end{proposition}

The following result implies the forward direction of Proposition~\ref{bchprop} (Proposition~$4.8$ from \cite{BCH22}) since \(\det(A) = -(2x+1)\), as pointed out in \cite{Heil23}. 

\begin{theorem}[\cite{Heil23}]\label{HmainThm}
Let \(A \in \Xi\). If \(\det(A) = mp\) where \(p\) is prime and \(m \in \{1,2,4\}\), then \(A\) is prime.
\end{theorem}

Our next proposition illustrates why we only need consider nonsingular matrices and provides a polynomial time procedure for determining if a unimodular matrix is prime (a matrix is unimodular if its determinant is \(\pm1\)). The proof proceeds by elementary calculations showing that if a matrix has determinant \(0, \pm 1\), then it is necessarily dominated and thus composite.

\begin{proposition}\label{prop:singular/unimodular matrices}
Let \(A=\big(\begin{smallmatrix} a_1 & a_2 \\ a_3 & a_4 \end{smallmatrix}\big) \in \mathbb{N}^{2 \times 2}\). If \(A\) is singular then it is composite. If \(A\) is unimodular, then it is composite, unless \(A\) is a unit or an associate of \(\big(\begin{smallmatrix} 1 & 1 \\ 0 & 1 \end{smallmatrix}\big)\), in which case it is prime.
\end{proposition}

\begin{proof}
We first consider the case that \(A\) is singular, thus \(\det(A) = 0\). We may assume that \(a_3\) is the smallest element of \(A\) without loss of generality by Equation~(\ref{assocEqn}). Since \(A\) is singular, then \(a_1a_4-a_2a_3 = 0\), implying that \(a_4 = \frac{a_2a_3}{a_1}\); noting that \(a_1 \neq 0\) since \(a_3\) is the smallest element of \(A\) and if \(a_1 = a_3 = 0\) then \(A\) has a left zero column and thus \(A = A\big(\begin{smallmatrix} 1 & 1 \\ 0 & 1 \end{smallmatrix}\big)\), and \(A\) is composite. Therefore \(A=\big(\begin{smallmatrix} a_1 & a_2 \\ a_3 & \frac{a_2a_3}{a_1} \end{smallmatrix}\big)\). Since \(a_3\) is the smallest element of \(A\), if \(a_1 \leq a_2\) or \(\frac{a_2a_3}{a_1} \leq a_2\), then the matrix is dominated, in which case it has a trivial factor \(\big(\begin{smallmatrix} 1 & 1 \\ 0 & 1 \end{smallmatrix}\big) \in \mathbb{N}^{2 \times 2}\). Since then \(\frac{a_2a_3}{a_1} > a_2\), it implies that \(\frac{a_3}{a_1} > 1\) which is a contradiction since \(a_3\) is minimal. Thus \(A\) is composite. 

Finally, consider the case where $A$ is unimodular, thus $\det(A)=\pm 1$.  We aim to show that if $A$ is unimodular, then either \(A\) is dominated or is a unit. We only consider the case when $\det(A)= 1$ since the case when $\det(A)= -1$ is analogous.

Suppose $a_1a_4-a_2a_3=1$. We first assume for contradiction that $\min(a_2, a_3)>\max (a_1, a_4)$. As $A\in\mathbb{N}^{2\times 2}$, then $\min(a_2, a_3)\geq \max(a_1, a_4)+1$. We can see that $a_2a_3\geq (\max(a_1,a_4)+1)^2=\max(a_1, a_4)^2+2\max(a_1, a_4)+1$ and $a_1a_4\leq \max(a_1, a_4)^2$. As $a_1a_4-a_2a_3=1$, then $a_2a_3=a_1a_4-1\leq \max(a_1,a_4)^2-1$, and it follows that  $\max(a_1, a_4)^2+2\max(a_1, a_4)+1\leq \max(a_1, a_4)^2-1$ and so $\max(a_1,a_4)\leq -1$, which contradicts the assumption that $A\in\mathbb{N}^{2\times 2}$, and so we must have $\min(a_2, a_3)\leq \max (a_1, a_4)$.

Next, we assume for contradiction that 
$\min(a_1, a_4)>\max(a_2, a_3)$. As $A\in\mathbb{N}^{2\times 2}$, then $\min(a_1, a_4)\geq \max(a_2,a_3)+1$. We can see that $a_1a_4\geq (\max(a_2, a_3)+1)^2=\max(a_2,a_3)^2+2\max(a_2,a_3)+1$ and $a_2a_3\leq \max(a_2,a_3)^2$.  As $a_1a_4-a_2a_3=1$, then $a_1a_4=a_2a_3+1\leq \max(a_2,a_3)^2+1$, and it follows that $\max(a_2, a_3)^2+2\max(a_2, a_3)+1\leq \max(a_2, a_3)^2+1$, and so $\max(a_2, a_3)\leq 0$. As $a_2, a_3\geq0$, it must follow that $a_2=a_3=0$ and so $a_1=a_4=1$. We then obtain that $A=I$, which is a unit. Otherwise, if \(A\) is not a unit, we have a contradiction with $A\in\mathbb{N}^{2\times 2}$ and so it must follow that $\min(a_1, a_4)\leq \max(a_2, a_3)$.

We have shown that if $A$ is not a unit, then $\min(a_2, a_3)\leq \max (a_1, a_4)$ and $\min(a_1, a_4)\leq \max(a_2, a_3)$, which implies that $A$ is a dominated matrix. Therefore $A$ can be factored by an associate of $\big(\begin{smallmatrix}
    1 & 1 \\ 0 & 1
\end{smallmatrix} \big)$. This means that $A$ is either composite or it is equal to one of these associates.
\end{proof}

The above results imply that the only matrices within \(\mathbb{N}^{2 \times 2}\) for which we cannot immediately determine if they are prime are elements in \(\Xi\) with composite determinants. To see this, note that in Proposition~\ref{BaethXiprop}, the only elements of \(\mathbb{N}^{2 \times 2}\) that are not completely characterised, are those with no zero elements, and such matrices must belong to \(\Xi\) to possibly be prime. Then, Theorem~\ref{HmainThm} says that matrices in \(\Xi\) are prime if they have  a prime determinant. Thus, the consideration of matrices within \(\Xi\) with composite determinant, and the provision of an efficient algorithm to determine primality is the main topic of this paper and considered in the next section.

\section{\(\Xi\)-matrices of composite determinant}

We now approach the problem of determining whether matrices which have a composite determinant and which belong to \(\Xi\) are prime, recalling that all other matrices from \(\mathbb{N}^{2 \times 2}\) are already characterised. For example, given \(A = \big(\begin{smallmatrix}  1382 & 1243 \\ 1045 & 1316 \end{smallmatrix}\big) \in \Xi\), how might we determine if \(A\) is prime without a priori knowledge of matrices \(B\) and \(C\), noting that \(\det(A) = 2^6\cdot 23 \cdot 31\)? Unfortunately, we do not know of any characterisation of such matrices. However, a brute force algorithm suggests itself, although it is not efficient.

\begin{proposition}\label{bruteForceProp}
Given a matrix \(A = \big(\begin{smallmatrix} a_1 & a_2\\ a_3 & a_4 \end{smallmatrix}\big) \in \Xi \subseteq \mathbb{N}^{2 \times 2}\) with a composite determinant, then determining if \(A\) is prime and computing a factorisation can be achieved in time \(\tilde{O}((\max_{1 \leq j \leq 4}{a_j})^4)\). 

Moreover, if we are given a factorisation \(\det(A) = pq\) with \(|p|, |q| \geq 2\), then determining if \(A = BC\) where \(B, C \in \mathbb{N}^{2 \times 2}\) and \(\det(B) = p\), \(\det(C) = q\), can be achieved in time \(\tilde{O}((\max_{1 \leq j \leq 4}{a_j})^3)\), and we can find such \(B, C\) in the same time bound.
\end{proposition}

\begin{proof}
    Consider a matrix \(A = \big(\begin{smallmatrix} a_1 & a_2\\ a_3 & a_4 \end{smallmatrix}\big) \in \Xi\). To determine whether \(A\) is prime, we need to check if there exist two non-unit matrices \(B, C \in \mathbb{N}^{2 \times 2}\) such that \(A = BC\). Suppose \(A = \big(\begin{smallmatrix} a_1 & a_2\\ a_3 & a_4 \end{smallmatrix}\big)\), \(B = \big(\begin{smallmatrix} b_1 & b_2\\ b_3 & b_4 \end{smallmatrix}\big)\), and \(C = \big(\begin{smallmatrix} c_1 & c_2\\ c_3 & c_4 \end{smallmatrix}\big)\). If  \(A = BC\), then we have
\[
\big(\begin{smallmatrix} a_1 & a_2\\ a_3 & a_4 \end{smallmatrix}\big) = \big(\begin{smallmatrix} b_1 & b_2\\ b_3 & b_4 \end{smallmatrix}\big) \cdot \big(\begin{smallmatrix} c_1 & c_2\\ c_3 & c_4 \end{smallmatrix}\big) =
\big(\begin{smallmatrix} b_1c_1\,+\,b_2c_3 & b_1c_2\,+\,b_2c_4\\ b_3c_1\,+\,b_4c_3 & b_3c_2\,+\,b_4c_4 \end{smallmatrix}\big).
\]
Since \(\det(A)\neq 0\), it follows that \(b_1\neq 0\) or \(b_3\neq 0\), and hence \(c_1 \leq \max\{a_1,a_3\} \leq \max_{1 \leq j \leq 4}{a_j}\). In a similar way, we can show that \(b_i, c_i\leq \max_{1 \leq j \leq 4}{a_j}\) for \(1\leq i\leq 4\).

To decide if \(A\) can be factorised as a product \(A = BC\) of non-units \(B, C \in \mathbb{N}^{2 \times 2}\), we need to check if there exists a matrix \(C\) such that: (i) it has non-zero determinant; (ii) it is not a unit; and (iii) \(AC^{-1}\) is a non-negative integer matrix that is also not equal to a unit. By the above observation, there are at most \((\max_{1 \leq j \leq 4}{a_j})^4\) possible candidates for \(C\). Note that checking the conditions (i)--(iii) can be done in time polynomial in \(\log (\max_{1 \leq j \leq 4}{a_j})\), which proves the first part of the proposition, since if we ever find such \(B, C\) we can then output them.

To show the second part, we need to check if there is a matrix \(C\) that in addition to~(i)--(iii) above also satisfies \(\det(C) = q\). Note that in this case, there are at most \((\max_{1 \leq j \leq 4}{a_j})^3\) possible candidates for \(C\) since given any three entries of \(C\), the remaining one will be determined by the constraint \(c_1c_4 - c_2c_3 = q\). 
\end{proof}

An alternative approach is given in \cite{Heil23}, where an algorithm is provided to find all \(\Xi\) matrices with a determinant equal to some \(k\), since such matrices are in a one-to-one relation with so-called `\emph{$*$-triples}' studied by Raney \cite{Raney1973OnCF}. However as the author of \cite{Heil23} notes, the algorithm becomes cumbersome for larger \(k\) and indeed the number of \(*\)-triples of \(\det(A)\) grows as fast as the product of \(\det(A)\) and the number of divisors of \(\det(A)\) and therefore such an approach seems prohibitive, although an exact complexity analysis is not given.

In order to find a more efficient procedure, we begin by considering a Smith normal form of elements of \(\Xi\). We define this lemma for some arbitrary nonsingular adj-coprime matrix \(X = \big(\begin{smallmatrix} x_{1} & x_{2} \\ x_{3} & x_{4} \end{smallmatrix}\big) \in \mathbb{N}^{2 \times 2}\) and identify some important properties of each element of the Smith normal form that we will require. The proof is given by considering simple properties of the Smith normal form of nonsingular adj-coprime \(2 \times 2\) nonnegative integer matrices and B\'ezout coefficients.  

\begin{lemma}\label{smithFormLemma}
Let \(X = \big(\begin{smallmatrix} x_{1} & x_{2} \\ x_{3} & x_{4} \end{smallmatrix}\big) \in \mathbb{N}^{2 \times 2}\) be a matrix with \(\gcd(x_{1}x_{4},\, x_{2}x_{3}) = 1\) and \(\det(X) \neq 0\). Then there exists a Smith normal form \(X = U_X^{-1} S_X V_X^{-1}\) given by
\[
U_X = \begin{pmatrix} s & t \\ -x_{3} & x_{1} \end{pmatrix},\
S_X = \begin{pmatrix} 1 & 0 \\ 0 & \det(X) \end{pmatrix},\
V_X = \begin{pmatrix} 1 & \sigma \\ 0 & 1 \end{pmatrix},
\]
where 
\(
s = a + kx_{3},\ t = b - kx_{1}, \text{ and } \sigma=-(ax_{2} + bx_{4}) + k\det(X),
\) 
for any \(k\in \mathbb{Z}\) and a fixed pair \(a, b \in \mathbb{Z}\) that can be computed in polynomial time such that \(ax_{1} + bx_{3} = 1\).

In particular, there exists a Smith normal form with \(0 \leq \sigma < |\det(X)|\).
Moreover, since \(X\) is non-negative, we also have \(x_{1}\sigma +t \det(X) \leq 0\) and \(-x_{3}\sigma + s \det(X) \geq 0\).
\end{lemma}

\begin{proof}
Since \(\gcd(x_1, x_3) = 1\), there exists a pair \(a, b \in \mathbb{Z}\) such that \(ax_1 + bx_3 = 1\), where \(a, b\) are of course the B\'ezout coefficients for \(x_1, x_3\). These coefficients are not unique, however there exist two such pairs \(a,b\) where \(|a| \leq |x_3|\) and \(|b| \leq |x_1|\). Given one such pair (which can be computed in polynomial time, for example, via the Euclidean algorithm), the set of all B\'ezout coefficients is then given by \(\{(a + kx_3, b - kx_1)\ | \ k \in \mathbb{Z}\}\). Hence we define \(s = a + kx_3\) and \(t = b - kx_1\), where \(k \in \mathbb{Z}\).

We see that \(\det(U_X) = \det(V_X) = 1\); thus \(U_X, V_X\) are unimodular, and 
\(\det(S_X) = \det(X)\) as required. Trivially, \(U_X^{-1} = \frac{1}{sx_1 + tx_3}\big(\begin{smallmatrix} x_1 & -t \\ x_3 & s\end{smallmatrix}\big) = \big(\begin{smallmatrix} x_1 & -t \\ x_3 & s\end{smallmatrix}\big)\) and \(V_X^{-1} = \big(\begin{smallmatrix} 1 & -\sigma \\ 0 & 1\end{smallmatrix}\big)\), since \(sx_1 + tx_3 = 1\). We therefore see that \(S^{}_XV_X^{-1} = \big(\begin{smallmatrix} 1 & -\sigma \\ 0 & \det(X)\end{smallmatrix}\big)\) and thus
\begin{eqnarray*}
U_X^{-1} S^{}_X V_X^{-1}   &=& \begin{pmatrix} x_1 & -t \\ x_3 & s \end{pmatrix} \begin{pmatrix} 1 & -\sigma \\ 0 & \det(X) \end{pmatrix} \\
  &=& \begin{pmatrix} x_1 & -(x_1\sigma +t \det(X)) \\ x_3 & -x_3\sigma + s \det(X) \end{pmatrix}.
\end{eqnarray*}
Note that
\begin{align*}
\sigma &= -(ax_2 + bx_4) + k\det(X)\\
&= -(ax_2 + bx_4) + k(x_1x_4 - x_2x_3)\\
&= -((a+kx_3)x_2 + (b-kx_1)x_4) = -(sx_2 + tx_4).
\end{align*}
Recalling that \(sx_1 + tx_3 = 1\) and \(\det(X) = x_1x_4-x_2x_3\), we see that 
\begin{eqnarray*}
    -(x_1\sigma +t \det(X)) & = & x_1(sx_2 + tx_4) - t(x_1x_4-x_2x_3) \\
    & = & x_2 (sx_1 + tx_3) = x_2.
\end{eqnarray*}
Similarly, 
\begin{eqnarray*}
-x_3\sigma + s \det(X) & = & x_3(sx_2 + tx_4) + s (x_1x_4-x_2x_3) \\
 & = & x_4(sx_1 + tx_3) = x_4,
\end{eqnarray*}
and thus \(U_X^{-1} S^{}_X V_X^{-1}= X\) as required. Finally, notice that \(X \in \mathbb{N}^{2 \times 2}\) implies \(x_2, x_4 \geq 0\) and hence \(x_1\sigma +t \det(X) \leq 0\) and \(-x_3\sigma + s \det(X) \geq 0\).
\end{proof}

\begin{example}
We give here an application of Lemma~\ref{smithFormLemma}. Consider a matrix \(X = \big(\begin{smallmatrix} 43 & 31 \\ 24 & 35 \end{smallmatrix}\big)\). Clearly then \(\gcd(43\cdot35, 31\cdot 24) = 1\) and thus we can apply the lemma. We give two different Smith normal forms of \(X\) according to whether \(\sigma\) is positive or negative (taking the minimal such positive or negative representatives):
\begin{eqnarray*}
U^-_X = \begin{pmatrix} -5 & 9 \\ -24 & 43 \end{pmatrix}, &
S^-_X = \begin{pmatrix} 1 & 0 \\ 0 & 761 \end{pmatrix}, &
V^-_X = \begin{pmatrix} 1 & -160 \\ 0 & 1 \end{pmatrix}; \\
U^+_X = \begin{pmatrix} 19 & -34 \\ -24 & 43 \end{pmatrix}, &
S^+_X = \begin{pmatrix} 1 & 0 \\ 0 & 761 \end{pmatrix}, &
V^+_X = \begin{pmatrix} 1 & 601 \\ 0 & 1 \end{pmatrix}.
\end{eqnarray*}
Notice that \(\det(C) = 761\), \(\sigma \equiv 601 \text{ mod } 761\), \(U^{\pm 1}_X\) and \(V^{\pm 1}_X\) are unimodular, and also that 
\(
X = (U_X^-)^{-1} S_X^- (V_X^-)^{-1} = (U_X^+)^{-1} S_X^+ (V_X^+)^{-1}
\), as required.
\end{example}

The next technical lemma is crucial to our approach. Consider that we are given \(A \in \Xi\) with a composite determinant, and thus \(A \in \mathbb{N}^{2 \times 2}\) is a non-singular non-dominated adj-coprime matrix. Then the lemma says that if there exists some factorisation \(A = BC\) with \(B, C \in \mathbb{N}^{2 \times 2}\) and we have a priori knowledge of the determinant of \(C\), without knowing \(C\) itself, then we can determine matrices \(V_C\) and \(S_C\) from the Smith normal form of \(C\). 

\begin{lemma}\label{Vclemma}
There exists a polynomial-time procedure that, given a matrix \(A \in \Xi\) with a composite determinant and integers  \(p,q\) such that \(|p|, |q| \geq 2\) and \(\det(A) = pq\), either
\begin{enumerate}[(i)]
\item determines that no matrices \(B, C \in \mathbb{N}^{2 \times 2}\) with \(\det(B) = p\) and \(\det(C) = q\) satisfy \(A = BC\), or

\item constructs a matrix \(S_C \in \mathbb{Z}^{2 \times 2}\) and \(\sigma \in \mathbb{N}\) such that if there are matrices \(B, C \in \mathbb{N}^{2 \times 2}\) satisfying \(\det(B) = p\), \(\det(C) = q\) and \(A = BC\), then \(C\) has a  Smith normal form \(U_C, S_C, V_C\) as in Lemma~\ref{smithFormLemma}, where  \(V_C = \big(\begin{smallmatrix} 1 & \sigma \\ 0 & 1 \end{smallmatrix}\big)\) and \(0 \leq \sigma < |\det(C)|\).
\end{enumerate}
\end{lemma}

\begin{proof}
Consider a matrix \(A \in \Xi\) and
suppose it has a factorisation \(A = BC\) for some matrices \(B, C \in \mathbb{N}^{2 \times 2}\), where \(A = \big( \begin{smallmatrix} a_1 & a_2 \\ a_3 & a_4 \end{smallmatrix} \big)\), \(B = \big( \begin{smallmatrix} b_1 & b_2 \\ b_3 & b_4 \end{smallmatrix} \big)\) and \(C = \big( \begin{smallmatrix} c_1 & c_2 \\ c_3 & c_4 \end{smallmatrix} \big)\). Then \(C\) must have the property \(\gcd(c_1,c_2,c_3,c_4) = 1\). Indeed, if there is \(d>1\) such that \(d\, |\, c_j\) for \(1\leq j\leq 4\), then \(d\, |\, a_j\) for \(1\leq j\leq 4\), which contradicts the assumption that \(A\) is adj-coprime.
Therefore, by Proposition \ref{SNFprop}, the matrix \(S_C\) in the Smith normal form of \(C\) has the form \(S_C = \big( \begin{smallmatrix} 1 & 0 \\ 0 & \det(C) \end{smallmatrix} \big)\). More difficult is to determine \(V_C = \big( \begin{smallmatrix} 1 & \sigma \\ 0 & 1 \end{smallmatrix} \big)\), since we do not know \(B\) or \(C\), only \(\det(C)\), and we need to determine the possible value of \(\sigma\) such that \(0 \leq \sigma < |\det(C)|\) under the assumption that \(A = BC\) holds for some \(B, C \in \mathbb{N}^{2 \times 2}\) with \(\det(B) = p\) and \(\det(C) = q\). 

Since \(A=BC\) and \(C\) is nonsingular (given that \(\det(A) \neq 0\)), we have \(AC^{-1} = B \in \mathbb{N}^{2 \times 2}\). Matrix \(C\) has a Smith normal form and can thus be written \(C = U_C^{-1}S_CV_C^{-1}\) with \(U_C, S_C, V_C \in \mathbb{Z}^{2 \times 2}\) and \(U_C, V_C\) are unimodular, 
and thus \(AC^{-1} = AV_CS_C^{-1}U_C = B\). This implies that \(AV_CS_C^{-1} = BU_C^{-1} \in \mathbb{Z}^{2 \times 2}\). Note that \(BU_C^{-1}\) must be an integer matrix but not necessarily nonnegative. We thus consider constraints to make \(AV_CS_C^{-1}\) integer, recalling that we know \(A\) and \(S_C\) and that \(V_C = \big( \begin{smallmatrix} 1 & \sigma \\ 0 & 1 \end{smallmatrix} \big)\) for some unknown \(\sigma \in \mathbb{Z}\).

To determine \(\sigma\), notice that
\[
AV_CS_C^{-1} = \begin{pmatrix} a_1 & \det(C)^{-1}(a_2 + \sigma a_1) \\ a_3 & \det(C)^{-1}(a_4 + \sigma a_3) \end{pmatrix}.
\]
In order that \(AV_CS_C^{-1} \in \mathbb{Z}^{2 \times 2}\), i.e., to ensure that it has integer rather than rational elements, we must therefore have that
\begin{eqnarray}
    a_2 + \sigma a_1 & \equiv & 0 \mod \det(C), \text{ and } \label{eeq1} \\
    a_4 + \sigma a_3 & \equiv & 0 \mod \det(C). \label{eeq2}
\end{eqnarray}
Each congruence has a unique solution if \(\gcd(a_1, \det(C)) = \gcd(a_3, \det(C)) = 1\). This is true as we now show. 
 Since \(\det(C)\, |\, \det(A)\), 
 we see that 
\[ \gcd(a_1, c_1c_4-c_2c_3) \leq \gcd(a_1, a_1a_4-a_2a_3).
\]
Assume by contradiction that \(\gcd(a_1, a_1a_4-a_2a_3) = d > 1\), then \(d | a_1\) and \(d | a_1a_4-a_2a_3\), thus
\(d | a_2a_3\) giving a contradiction since \(\gcd(a_1a_4,a_2a_3) = 1\) by definition of \(\Xi\). Therefore, \(\gcd(a_1, c_1c_4-c_2c_3) = \gcd(a_1, a_1a_4-a_2a_3) = 1\).

If both congruences yield the same solution \(\bmod \det(C)\), then we have determined a unique \(\sigma\) such that \(0 \leq \sigma < |\det(C)|\), otherwise no \(C \in \mathbb{N}^{2 \times 2}\) exists satisfying the constraint that \(B\) is an integer matrix.

Finally, notice that we can decide in polynomial time whether the system of congruences (\ref{eeq1}) and (\ref{eeq2}) has a solution. Hence the whole procedure runs in polynomial time.
\end{proof}

We now move to the main theorem of the paper, which states that if we are given a matrix \(A \in \Xi\) and a value \(c \in \mathbb{N}\), then we can determine if it can be factored as \(A=BC\) with \(B, C \in \mathbb{N}^{2 \times 2}\) and \(\det(C) = c\) in time \(\tilde{O}(\mu(A))\), i.e., in time that is linear in the minimal element of \(A\). In order to prove this, we will make use of the following lemma which states that our assumption that \(c > 0\) does not actually restrict the generality and that neither \(B\) nor \(C\) contains any zero elements. This will simplify several parts of the main theorem.

\begin{lemma}\label{posDetNonNegLem}
Let \(A = \big(\begin{smallmatrix} a_1 & a_2 \\ a_3 & a_4 \end{smallmatrix}\big) \in \Xi\) and let \(c \in \mathbb{Z}\) be a value satisfying \(|c| \leq |\det(A)|\). If there exists a factorisation \(A = BC\) with \(B = \big(\begin{smallmatrix} b_1 & b_2 \\ b_3 & b_4 \end{smallmatrix}\big) \in \mathbb{N}^{2\times 2}\) and \(C = \big(\begin{smallmatrix} c_1 & c_2 \\ c_3 & c_4 \end{smallmatrix}\big) \in \mathbb{N}^{2\times 2}\) such that \(|\det(C)| = |c|\), then
\begin{itemize}
\item we may assume without loss of generality that \(\det(C) > 0\), and
\item if \(|c| < |\det(A)|\) then \(B\) and \(C\) must be strictly positive, i.e., \(b_j, c_j > 0\) for \(1 \leq j \leq 4\).
\end{itemize}
\end{lemma}

\begin{proof}
    Under the assumptions of the lemma, consider the case when \(\det(C) < 0\). Let \(J = \big(\begin{smallmatrix} 0 & 1 \\ 1 & 0 \end{smallmatrix}\big)\). Clearly \(J^2 = I\). Therefore, we can write
\[
A = BC = (BJ)(JC),
\]
where we note that \(\det(JC) = -\det(C)\) since \(\det(J) = -1\). 
Therefore, we may assume that \(C\) has a positive determinant as left multiplying \(C\) by \(J\) swaps the rows of \(C\) and does not change the sign of its elements.

Next we want to prove that both \(B\) and \(C\) are necessarily strictly positive if \(|c| < |\det(A)|\). To see this, assume that \(b_4 = 0\). We will show that in this case \(b_3 = 1\). Indeed, if \(b_3 = 0\), then \(B\) and hence \(A\) will be singular. On the other hand, and if \(b_3 > 1\), then we can factorise \(B\) as
\[
\left(\begin{array}{cc} b_1 & b_2 \\ b_3 & 0 \end{array}\right) = \left(\begin{array}{cc} 1 & 0 \\ 0 & b_3 \end{array}\right) \left(\begin{array}{cc} b_1 & b_2 \\ 1 & 0 \end{array}\right),
\]
which implies that \(A = \big(\begin{smallmatrix} 1 & 0 \\ 0 & b_3 \end{smallmatrix}\big)B'C\) for some \(B', C \in \mathbb{N}^{2 \times 2}\). This implies that \(A\) is not adj-coprime since \(b_3 > 1\) and \(b_3 | \gcd(a_3,a_4)\). Therefore, \(b_3 = 1\), and the matrix \(A\) can be written as
\[
A = BC = \left(\begin{array}{cc} b_1 & b_2 \\ 1 & 0 \end{array}\right) \left(\begin{array}{cc} c_1 & c_2 \\ c_3 & c_4 \end{array}\right) = 
\left(\begin{array}{cc} b_1c_1+b_2c_3 & b_1c_2+b_2c_4 \\ c_1 & c_2 \end{array}\right).
\]
The top row of this matrix is larger than the bottom if \(b_1>0\), implying \(A\) is dominated and thus \(A \not\in \Xi\). We therefore must have \(b_1=0\) and \(A = \big(\begin{smallmatrix} b_2c_3 & b_2c_4 \\ c_1 & c_2 \end{smallmatrix}\big)\). Note that \(b_2 | \gcd(a_1,a_2)\), and hence \(b_2=1\) by the adj-coprime property of \(A\). This implies that \(B = J\) and \(|c| = |\det(A)|\), contradicting our assumption.

The other three cases (when \(b_1, b_2, b_4 = 0\)) are similar, deriving first that the other element on the same row must be \(1\) (since otherwise the matrix is either singular or else can be factored by one of \(\{\big(\begin{smallmatrix} 1 & 0 \\ 0 & b_4 \end{smallmatrix}\big), \big(\begin{smallmatrix} b_1 & 0 \\ 0 & 1 \end{smallmatrix}\big), \big(\begin{smallmatrix} b_2 & 0 \\ 0 & 1 \end{smallmatrix}\big)\}\), thus proving that \(A\) is not adj-coprime), and then reasoning that a row composed of a \(0\) and a \(1\) implies that \(A\) is dominated.

Finally we must prove that \(C\) is also strictly positive. Rather than going over four additional cases, we may recall that \(A = BC\) implies that \(A^T = C^T B^T\). Note that transposing \(A\) does not affect its primality and that \(\Xi\) is invariant under transposition. Hence by using the same argument, we can show that \(C\) does not have any zero elements.
\end{proof}

We are ready to prove our main result, now knowing that the determinant of matrix \(C\) can be assumed positive.

\begin{theorem}\label{singleDetLem}
    Given a matrix \(A = \big(\begin{smallmatrix} a_1 & a_2 \\ a_3 & a_4 \end{smallmatrix}\big) \in \Xi\) and \(c \in \mathbb{N}\) such that \(c\, | \det(A)\), we can determine  whether there exist non-unit matrices \(B, C \in \mathbb{N}^{2 \times 2}\) such that \(\det(C) = c\) and \(A=BC\), and we can compute such (non-unique) \(B, C\) if they exist, in time \(\tilde{O}(\mu(A))\).
\end{theorem}
\begin{proof}
 {\bf Dealing with case \(c = 1\) or \(c = \det(A)\)}.  Firstly, we show that if \(c=1\) or \(c=\det(A)\), then \(A\) has only a trivial factorisation of the required form, that is, if \(A = BC\) for some matrices \(B, C \in \mathbb{N}^{2 \times 2}\) with \(\det(C) = c\), then \(C=I\) or \(C=A\), respectively. Consider the case when \(c=\det(A)\) since the other case is analogous. Suppose \(A = BC\), where \(B, C \in \mathbb{N}^{2 \times 2}\) and \(\det(C) = \det(A)\); thus \(\det(B) = 1\). If \(B\neq I\), then by Theorem~\ref{detOneAtoms} it is equal to a product of primes \(
P = \big(\begin{smallmatrix} 1 & 0 \\ 1 & 1 \end{smallmatrix}\big) \text{ and }
Q = \big(\begin{smallmatrix} 1 & 1 \\ 0 & 1 \end{smallmatrix}\big)
\). Hence \(A\) can be written as a product \(A = PB'C\) or \(A = QB'C\) for some \(B' \in \mathbb{N}^{2 \times 2}\). It is not hard to see that in this case \(A\) is dominated, a contradiction with the assumption that \(A \in \Xi\). Thus we assume for the remainder of the proof that \(1 < c < \det(A)\).

\noindent {\bf Reducing the number of unknowns in \(B\) to three}. We may assume that \(\mu(A) = a_3\) by Equation~(\ref{assocEqn}). Let \(U_C, S_C, V_C\) denote a Smith normal form of (unknown) matrix \(C\) such that \(U_C\) and  \(V_C\) are unimodular, \(U_C, S_C, V_C \in \mathbb{Z}^{2 \times 2}\) and \(C = U_C^{-1}S_C V_C^{-1}\). Since we are given \(A \in \Xi\) and \(c = \det(C)\), if there exists \(A = BC\) with \(B \in \mathbb{N}^{2 \times 2}\) and \(\det(C) = c\), then we can determine \(S_C\) and \(V_C\) by Lemma~\ref{Vclemma}. We therefore have:
\[
U_C = \begin{pmatrix} s & t \\ -c_3  & c_1  \end{pmatrix}, \quad
S_C = \begin{pmatrix} 1 & 0 \\ 0 & \det(C) \end{pmatrix}, \quad
V_C = \begin{pmatrix} 1 & \sigma \\ 0 & 1 \end{pmatrix},
\]
with \(0 \leq \sigma < |\det(C)|\), recalling that \(U_C\) remains unknown. Since \(U_C\) is unimodular, then \(U_C = \left(\begin{smallmatrix} s & -\frac{(c_1 s - 1)}{c_3 } \\ -c_3  & c_1  \end{smallmatrix}\right)\). We may therefore derive
\begin{eqnarray}
C = U_C^{-1}S_C V_C^{-1} = \begin{pmatrix} c_1  & \frac{\det(C)(c_1 s - 1)}{c_3 } - c_1 \sigma \\ c_3  & \det(C)s - c_3 \sigma \end{pmatrix}. \label{cForm}
\end{eqnarray}

Thus we see that we can determine some details of \(B = AC^{-1}\) as follows
\[
B = AV_CS_C^{-1}U_C =
\left(\begin{array}{cc}
a_1 \,s-\frac{c_3  \,{\left(a_2 +a_1 \,\sigma\right)}}{c} & \frac{c_1  \,{\left(a_2 +a_1 \,\sigma\right)}}{c}-\frac{a_1 \,{\left(c_1  \,s-1\right)}}{c_3  }\\
a_3 \,s-\frac{c_3  \,{\left(a_4 +a_3 \,\sigma\right)}}{c} & \frac{c_1  \,{\left(a_4 +a_3 \,\sigma\right)}}{c}-\frac{a_3 \,{\left(c_1  \,s-1\right)}}{c_3  }
\end{array}\right).
\]

It can be shown that \(a_2 + \sigma a_1 \equiv 0 \textrm{ mod } c\) and \(a_4 + \sigma a_3 \equiv 0 \textrm{ mod } c\), see the proof of Lemma~\ref{Vclemma}. 
Since we know \(a_1, a_2, a_3, a_4, c\) and \(\sigma\),
we may therefore define \(d_1 = \frac{a_2 + \sigma a_1}{c}\) and \(d_2 = \frac{a_4 + \sigma a_3}{c}\) noting that \(d_1, d_2 \in \mathbb{N}\) since \(0 \leq \sigma < |\det(C)|\) so that 
\begin{eqnarray}
B = AC^{-1} = AV_CS_C^{-1}U_C =
\left(\begin{array}{cc}
a_1 \,s-c_3 d_1 & c_1 d_1-\frac{a_1 \,{\left(c_1  \,s-1\right)}}{c_3  }\\
a_3 \,s-c_3 d_2 & c_1 d_2-\frac{a_3 \,{\left(c_1  \,s-1\right)}}{c_3  }
\end{array}\right).\label{derivation1}
\end{eqnarray}

It is useful to note at this point that the above matrix contains only three unknowns, namely \(s, c_1 \), and \(c_3 \). The rest of the proof focuses on reducing the number of unknowns to \emph{one} in the sense that once a `guess' is made for \(c_3\), then \(s\) and \(c_1\) are uniquely determined, allowing us to thus iterate over all possible values of \(c_3\). Recall that \(s \in \mathbb{Z}\) and \(a_{i}, c_1 , c_3 , d_1, d_2 \in \mathbb{N}\) for \(1 \leq i \leq 4\). If values can be found for the unknowns such that \(AC^{-1} = B \in \mathbb{N}^{2 \times 2}\) and \(C \in \mathbb{N}^{2 \times 2}\) then we have found a solution. We thus need \(AC^{-1}\) and \(C\) to be  \emph{integer} and \emph{nonnegative} (or in fact positive by Lemma~\ref{posDetNonNegLem}). 

\noindent {\bf Reducing the number of unknowns in \(B\) to one -- determining \(s\)}. Denote \(B = \big(\begin{smallmatrix} b_1 & b_2 \\ b_3 & b_4 \end{smallmatrix}\big)\). Since \(A = BC\) and \(\mu(A) = \min\{a_1, a_2, a_3, a_4\} = a_3\), then we note that \(a_3 = b_3c_1  + b_4c_3 \) and thus \(1\leq c_1 , c_3  \leq a_3 = \mu(A)\), recalling by Lemma~\ref{posDetNonNegLem} that \(b_j > 0\) and \(c_j > 0\) for \(1 \leq j \leq 4\). Therefore we iterate over each \(1 \leq c_3  \leq \mu(A)\). Fixing some such \(c_3\), in order that the left hand column of \(AC^{-1}\) is positive in Equation~(\ref{derivation1}), we have the constraints \(a_1 \,s-c_3 d_1 > 0\) and \(a_3 \,s-c_3 d_2 > 0\) with \(s\) the only unknown, and thus 
\begin{eqnarray}
s & > & \min\left(\frac{c_3 d_1}{a_1}, \frac{c_3 d_2}{a_3}\right). \label{minEq}
\end{eqnarray}
Any such choice of \(s\) satisfying this inequality yields a nonnegative integer left column of \(AC^{-1}\) and thus we consider constraints to ensure the same of the right column, recalling that \(c_3 \) is fixed. Since \(c_{i}, a_{i}, d_j \geq 0\) for \(1 \leq i \leq 4\) and \(1 \leq j \leq 2\), determining if a solution is possible for this choice of \(c_3 \) suggests that we should choose the smallest possible \(s\) according to the constraint of Equation~(\ref{minEq}), in order that the right column becomes nonnegative. However we must also satisfy that \(a_1 \left(c_1  \,s-1\right) \equiv 0 \text{ mod } c_3 \) and \(a_3 \left(c_1  \,s-1\right) \equiv 0 \text{ mod } c_3 \) in order that the elements of the right column are \emph{integers}. Let \(\delta_1 = \gcd(c_3 , a_1)\) and \(\delta_2 = \gcd(c_3 , a_3)\). Then define \(c'_{3} = \frac{c_3 }{\delta_1} \in \mathbb{N}\) and \(a'_{1} = \frac{a_1}{\delta_1} \in \mathbb{N}\) and \(c''_{3} = \frac{c_3 }{\delta_2} \in \mathbb{N}\) and \(a'_{3} = \frac{a_3}{\delta_2} \in \mathbb{N}\). Thus
\begin{eqnarray*}
a_1 c_1  \,s & \equiv & a_1 \text{ mod } c_3  \, \Rightarrow \, a'_{1} c_1  \,s  \equiv a'_{1} \text{ mod } c'_{3} \, \Rightarrow \, c_1  \,s \equiv 1 \text{ mod } c'_{3}, \\
a_3 c_1  \,s & \equiv & a_3 \text{ mod } c_3  \, \Rightarrow \, a'_{3} c_1  \,s  \equiv  a'_{3} \text{ mod } c''_{3} \, \Rightarrow \, c_1  \,s \equiv 1 \text{ mod } c''_{3}. 
\end{eqnarray*}
Thus \(c_1 s \equiv 1 \text{ mod } \text{lcm}(c'_{3}, c''_{3})\) and both \(c_1 \) and \(s\) are invertible modulo \(\text{lcm}(c'_{3}, c''_{3})\) since
\(
\gcd(c_1 , \text{lcm}(c'_{3}, c''_{3})) \leq \gcd(c_1 , c'_{3}c''_{3}) \leq \gcd(c_1 , c_3 c_3 )  = 1,
\) 
with the last equality following since \(c_1\) and \(c_3\) are coprime. Indeed, if this is not the case, then assume \(\gcd(c_1, c_3) = q > 1\) and thus 
\(C = \left(\begin{array}{cc} c_1 & c_2 \\ c_3 & c_4 \end{array}\right) = C'Q' := \left(\begin{array}{cc} \frac{c_1}{q} & c_2 \\ \frac{c_3}{q} & c_4 \end{array}\right)\left(\begin{array}{cc} q & 0 \\ 0 & 1 \end{array}\right),\)
which implies that \(A = BC = BC'Q' = A'Q'\) for some \(A' \in \mathbb{N}^{2 \times 2}\) and thus \(A\) is not adj-coprime and \(A \not\in \Xi\) contradicting our assumption. 
We can thus derive that \(c_1  \equiv s^{-1} \text{ mod } \text{lcm}(c'_{3}, c''_{3})\). We know that \(c_1  \geq 0\) since \(C \in \mathbb{N}^{2 \times 2}\). 

In order that \((AC^{-1})_{2,2} > 0\), \(c_1 d_2-\frac{a_3 \,{\left(c_1  \,s-1\right)}}{c_3 }>0\). Recall that \(s > \frac{c_3 d_2}{a_3}\), and hence \(s \geq \left\lceil\frac{c_3 d_2}{a_3}\right\rceil\) since \(s\) is an integer. We now show that if \(s > \left\lceil\frac{c_3 d_2}{a_3}\right\rceil\) then \((AC^{-1})_{2,2} \leq 0\). To see this, set \(s = \left\lceil\frac{c_3 d_2}{a_3}\right\rceil + 1\) and then note that
 \begin{eqnarray}
(AC^{-1})_{2,2} & = & c_1 \left(d_2 - \frac{a_3s}{c_3 }\right) + \frac{a_3}{c_3 } \quad\left[\text{let } s = \left\lceil\frac{c_3 d_2}{a_3}\right\rceil + 1 \text{ with } \frac{c_3 d_2}{a_3} + 1\right]\\
 & \leq & c_1 \left(d_2 - \frac{a_3c_3 d_2}{a_3c_3 } - \frac{a_3}{c_3 }\right) + \frac{a_3}{c_3 } = \frac{a_3}{c_3 } (1-c_1 ) \leq 0.
\end{eqnarray}
The last inequality holds since \(c_1 > 0\). To see this, note that \(C \in \mathbb{N}^{2 \times 2}\) and thus \(c_1 \geq 0\), and if \(c_1 = 0\), then 
\(C = \left(\begin{array}{cc} 0 & c_2 \\ c_3 & c_4 \end{array}\right) = C''Q'' := \left(\begin{array}{cc} 0 & c_2 \\ 1 & c_4 \end{array}\right)\left(\begin{array}{cc} c_3 & 0 \\ 0 & 1 \end{array}\right),\) 
in which case \(A = BC = BC''Q'' = A''Q''\) for some \(A'' \in \mathbb{N}^{2 \times 2}\) and thus \(A\) is not adj-coprime and \(A \not\in \Xi\) contradicting our assumption. Now, if \(c_1 = 1\), then \((AC^{-1})_{2,2} = b_4 = 0\), but this is a contradiction since \(b_j > 0\) for \(1 \leq j \leq 4\) as proven in Lemma~\ref{posDetNonNegLem}. Therefore, once \(c_3 \) is defined, then \(s\) is \emph{uniquely} determined in a possible solution as \(s = \left\lceil\frac{c_3 d_2}{a_3}\right\rceil\). 

We may apply a similar derivation by using the fact that \((AC^{-1})_{1,2} > 0\) and thus determine that \(s = \left\lceil\frac{c_3 d_1}{a_1}\right\rceil\).  Note then that \(\left\lceil\frac{c_3 d_2}{a_3}\right\rceil = \left\lceil\frac{c_3 d_1}{a_1}\right\rceil\), otherwise no solution exists since then 
one of the elements of \(AC^{-1}\) is necessarily negative.

\noindent {\bf A simpler expression for matrix \(B\)}. We may then denote \(s = \left\lceil\frac{c_3 d_1}{a_1}\right\rceil = \frac{c_3 d_1}{a_1} + \varepsilon_1\) and \(s = \left\lceil\frac{c_3 d_2}{a_3}\right\rceil = \frac{c_3 d_2}{a_3} + \varepsilon_2\) where \(0 \leq \varepsilon_1, \varepsilon_2 < 1\). 
Thus we may write Equation~(\ref{derivation1}) as:
{\footnotesize \begin{eqnarray} AC^{-1} & = & 
    \left(\begin{array}{cc}
a_1 \,s-c_3 d_1 & c_1 d_1-\frac{a_1 \,{\left(c_1  \,s-1\right)}}{c_3  }\\
a_3 \,s-c_3 d_2 & c_1 d_2-\frac{a_3 \,{\left(c_1  \,s-1\right)}}{c_3  }
\end{array}\right) \nonumber \\ & = & 
\left(\begin{array}{cc}
a_1 \!\left(\frac{c_3 d_1}{a_1} + \varepsilon_1\right)-c_3 d_1 & c_1 d_1-\frac{a_1 \!{\left(c_1  \!\left(\frac{c_3 d_1}{a_1} + \varepsilon_1\right)-1\right)}}{c_3  }\\
a_3 \!\left(\frac{c_3 d_2}{a_3} + \varepsilon_2\right)-c_3 d_2 & c_1 d_2-\frac{a_3 \!{\left(c_1  \!\left(\frac{c_3 d_2}{a_3} + \varepsilon_2\right)-1\right)}}{c_3  }
\end{array}\right) = 
\left(\begin{array}{cc}
a_1 \,\varepsilon_1 & \frac{a_1(1-c_1 \varepsilon_1)}{c_3 }  \\
a_3 \,\varepsilon_2 & \frac{a_3(1-c_1 \varepsilon_2)}{c_3 } 
\end{array}\right).\label{bForm}
\end{eqnarray}}
Note that \(\varepsilon_1, \varepsilon_2\) are dependent upon \(c_3 \). In fact, \(\varepsilon_1 = 1 - \frac{c_3 d_1 \text{ mod }{a_1}}{a_1}\) and \(\varepsilon_2 = 1 - \frac{c_3 d_2 \text{ mod }{a_3}}{a_3}\). 

To conclude then, we may note that we have derived matrices \(B\) and \(C\) in Equations~(\ref{bForm}) and (\ref{cForm}) respectively to have the following forms:
\[
B = \left(\begin{array}{cc}
a_1 \,\varepsilon_1 & \frac{a_1(1-c_1 \varepsilon_1)}{c_3 }  \\
a_3 \,\varepsilon_2 & \frac{a_3(1-c_1 \varepsilon_2)}{c_3 } 
\end{array}\right), \quad
C = \begin{pmatrix} c_1  & \frac{\det(C)(c_1 s - 1)}{c_3 } - c_1 \sigma \\ c_3  & \det(C)s - c_3 \sigma \end{pmatrix}.
\]
Note that \(a_1, a_3, \det(C), \sigma \in \mathbb{N}\) are fixed, with \(\det(C)\) given as in the theorem statement and \(\sigma\) derivable by Lemma~\ref{Vclemma}.  We may observe  that for each choice of \(1 \leq c_3  \leq \mu(A)\), then \(s = \left\lceil\frac{c_3 d_2}{a_3}\right\rceil  \in \mathbb{N}\) is uniquely determined as previously derived, as are \(\varepsilon_1, \varepsilon_2 \in \mathbb{Q}_{+}\), noting also that \(a_1\varepsilon_1 \in \mathbb{N}\) and \(a_3\varepsilon_2 \in \mathbb{N}\). The only remaining free variable is thus \(c_1 \). We must thus determine, for a particular value of \(1 \leq c_3  \leq \mu(A)\) with \(c_3  \in \mathbb{N}\), if there exists a value of \(c_1  \geq 0\) for which \(B, C \in \mathbb{N}^{2 \times 2}\). This requires \(B\) and \(C\) be both nonnegative and integer. 

\noindent {\bf Enforcing nonnegativity}. For nonnegativity, we may note informally that \(B\) tends to be nonnegative when \(c_1  \geq 0\) is small, whereas \(C\) tends to be nonnegative when \(c_1  \geq 0\) is large, giving a range of values of \(c_1 \) for which both \(B\) and \(C\) are nonnegative (we ignore integrality for the moment). More formally, let \(\xi_1 = \left\lfloor{\min{\left\{\frac{a_1}{\varepsilon_1}, \frac{a_3}{\varepsilon_2}\right\}}}\right\rfloor\) and note that \(b_j \geq 0\) for \(1 \leq j \leq 4\) and \(0 \leq c_1  \leq \xi_1\). Similarly, in order that \(C\) be nonnegative, define \(\xi_2 = \left\lceil \frac{\det(C)}{\det(C)s-\sigma c_3 }\right\rceil\). Thus, for each \(c_3  \in \mathbb{N}\), both \(B\) and \(C\) are nonnegative if \(\xi_1 \leq c_1  \leq \xi_2\). Note that this range can be empty, implying that both matrices cannot simultaneously be made nonnegative for the current choice of \(c_3 \).

\noindent {\bf Enforcing integrality}. We finally determine which values of \(c_1  \geq 0\) give integer \(B, C\) matrices, ignoring nonnegativity for the moment. There are only three elements of \(B\) and \(C\) which can be rational, specifically, \(b_2, b_4\), and \(c_2\). Each such entry gives a linear congruence relation for \(c_1\) modulo \(c_3\), recalling again that \(a_1\varepsilon_1, a_3\varepsilon_2 \in \mathbb{N}\). For example \(b_2 \in \mathbb{Z}\) implies that
\(
 \frac{a_1(1-c_1 \varepsilon_1)}{c_3 } \in \mathbb{Z}\) and therefore \((a_1\varepsilon_1) c_1 \equiv a_1 \mod c_3 .\)

We can decide in polynomial time if all three congruences have a common solution. If no solution exists, the current choice of \(c_3\) does not allow \(B\) and \(C\) to be simultaneously integer.

Otherwise, we can determine in polynomial time if this system has a solution \(c_1  \in \mathbb{N}\) satisfying \(\xi_1 \leq c_1  \leq \xi_2\). If the answer is positive, then we have found a solution pair \((c_1 , c_3 )\) for which \(B, C \in \mathbb{N}^{2 \times 2}\). Otherwise, we move on to the next value of \(c_3\). Since checking if \(c_1 \) exists for a given \(c_3 \) takes polynomial time, our procedure runs in time \(\tilde{O}(\mu(A))\).
\end{proof}

We have now proven our main result: given \(A \in \Xi \subseteq \mathbb{N}^{2 \times 2}\) with a composite determinant and a guess at the (positive, due to Lemma~\ref{posDetNonNegLem}) determinant of \(\det(C)\) such that \(\det(A) = \det(B)\det(C)\), then we can determine if nonunit matrices \(B, C \in \mathbb{N}^{2 \times 2}\) exist (and compute one such pair) where \(A=BC\) in time \(\tilde{O}(\mu(A))\). Of course the natural question is then to consider what happens when the guess at \(\det(C)\) is \emph{not} given. However we note that even in the $1$-dimensional case, given \(a \in \mathbb{N}\), factoring \(a\) is not a trivial task, and is not known to be possible in P; indeed the fastest algorithms require super-polynomial (but sub-exponential) time complexity (although determining if a particular integer is prime is in P \cite{agrawal2004primes}). Since our first consideration may be to at least factor the determinant of \(A \in \mathbb{N}^{2 \times 2}\), it seems natural that a polynomial-time algorithm does not appear easy. We can of course apply Theorem~\ref{singleDetLem} for each possible determinant of \(C\) such that \(\det(A) = \det(B)\det(C)\) which gives an efficient procedure for each such choice of determinant, which provides the following corollary. 

\begin{corollary}\label{finalCor}
    Given a matrix \(A \in \Xi\),
     we can determine if \(A\) is prime in time \(O(\gamma(A)^{\frac{2}{5}}) + \tilde{O}\left(\mu(A) \sqrt{e^{(1+o(1))\sqrt{\ln \gamma(A) / \ln \ln \gamma(A)}}}\right)\), where \(\mu(A)\) and \(\gamma(A)\) denote the minimal and maximal elements of \(A\) respectively.
\end{corollary}
\begin{proof}
If \(A \in \Xi\) is not prime, then \(A = BC\), where \(B, C \in \mathbb{N}^{2 \times 2}\) are non units, i.e., different from \(I\) or \(J = \big(\begin{smallmatrix} 0 & 1 \\ 1 & 0 \end{smallmatrix}\big)\). The factorization of an integer \(x\) can be found in \(O(x^{\frac{1}{5}})\) \cite{Ha21}. Clearly \(\det(A) \leq \gamma(A)^2\) and thus we can find the factorization of \(\det(A)\) in \(O(\gamma(A)^{\frac{2}{5}})\), giving the first term of the complexity bound.

Denote by \(d(\det(A))\), the number of \emph{divisors} of \(\det(A)\). It is known that the number of divisors of  any integer \(x\) is bounded by \(O\left(e^{(1+o(1))\sqrt{\ln x / \ln \ln x}}\right)\) \cite{Wi07}. We therefore find the factorization of \(\det(A)\) and then iterate over all possible divisors, applying Theorem~\ref{singleDetLem} for each such choice of \(\det(C)\). In fact we may assume that \(\det(C) < \det(B)\) since we note that \(X\) is prime if and only if \(X^T\) is prime and \(X = BC\) implies \(X^T = C^T B^T\), therefore giving the second term of the stated complexity bound. 

We may note then that our algorithm is particularly efficient when the number of divisors of \(\det(A)\), or else \(\mu(A)\), is small. 
\end{proof}

\section{Conclusion}
In this paper we provide a first efficient algorithm to determine if a matrix from \(\mathbb{N}^{2 \times 2}\) is prime and to find a non-trivial factorisation if the matrix is composite. There are links between this work and \emph{incidence matrices for morphisms} \cite{Honkala24}.

We may note that our main results, i.e., the derivation  of time complexity \(\tilde{O}(\mu(A))\) in Theorem~\ref{singleDetLem} alongside Corollary~\ref{finalCor} for proving primality, are still exponential in the (binary) representation size of \(A\). Determining whether the problem can be solved in time polynomial in the representation size of \(A\), or else proving that the problem is NP-hard is a logical next step, alongside identifying if there is  way to circumvent the requirement to factorise \(\det(A)\). Recall again that determining if an integer is composite is achievable in polynomial time in the binary representation of the integer \cite{agrawal2004primes}, however our setting is non-abelian, which intuitively appears to make the problem more difficult.

\bibliographystyle{splncs04}
\bibliography{references}

\end{document}